\documentclass[twocolumn, final]{elsarticle}

\usepackage[latin1]{inputenc}
\usepackage[english]{babel}
\usepackage{amsmath,amssymb}
\usepackage{graphicx}
\usepackage{hyperref}
\usepackage{verbatim}
\usepackage{color}

\newtheorem{lemma}{Lemma}
\newtheorem{assumption}{Assumption}

\newtheorem{defi}{Definition}

\newtheorem{exx}{Example}
\newtheorem{remm}{Remark}

\newenvironment{definition}{\begin{defi}\rm }{\end{defi}}
\newenvironment{proof}{\noindent {\em Proof.}}{\hfill \hspace*{1pt} \hfill $\square$}
\newenvironment{remark}{\begin{remm}\rm }{\hfill \hspace*{1pt} \hfill $\lrcorner$\end{remm}}

\newcommand\real{\ensuremath{{\mathbb R}}}

\newcommand{\ball}{\mathbb{B}}

\newcommand{\dom}{\mathrm{dom}\,}

\graphicspath{{figures/}}

\begin{document}

\begin{frontmatter}

\title{Hierarchical stability of nonlinear hybrid systems}

\author[roma]{Mario Sassano}\ead{mario.sassano@uniroma2.it}
\author[toulouse,trento]{Luca Zaccarian}\ead{zaccarian@laas.fr}

\address[roma]{Dipartimento di Ingegneria Civile e Ingegneria Informatica,
``Tor Vergata'', Via del Politecnico 1,
00133, Rome, Italy.\\}
\address[toulouse]{CNRS, LAAS, 7 avenue du colonel Roche, F-31400 Toulouse,
Univ. de Toulouse, LAAS, F-31400 Toulouse, France.\\}
\address[trento]{Dipartimento di Ingegneria Industriale, University of Trento, Italy.}

\begin{abstract}
In this short note we prove a hierarchical stability result that applies to hybrid dynamical systems satisfying the hybrid basic conditions of (Goebel et al., 2012). In particular, we 
establish sufficient conditions for uniform asymptotic stability of a compact set
based on some hierarchical stability assumptions involving two nested closed sets containing such a compact set. 
Moreover, mimicking the well known result for cascaded systems, we prove that the basin of attraction of such compact set coincides with the largest set from which all solutions are bounded.
The result appears to be useful when applied to several recent works involving hierarchical control architectures.
\end{abstract}

\end{frontmatter}

    \thispagestyle{empty}
    \pagestyle{empty}

\renewcommand{\thefootnote}{\arabic{footnote}}


\section{Introduction}
\label{sec:intro}

It is well known (see \cite{Sontag89CDC,Seibert90} for the global 
result or \cite[Theorem
3.1]{Vidyasagar80} for the local result) that for a cascade
interconnection of two nonlinear continuous-time systems, global asymptotic stability (GAS) of the origin of
the upper subsystem
and global asymptotic stability with zero input (0-GAS) of the origin of the 
lower subsystem implies local asymptotic stability (LAS) of the origin of the cascade, with
domain of attraction coinciding with the set of initial conditions
from which all trajectories are bounded. 
In particular, denoting by $x_1$ and $x_2$, respectively, the state
of the upper and lower subsystems, GAS of the closed set 
${\mathcal M} = \{(x_1,x_2):  x_1=0\}$ plus GAS of the
origin for initial conditions restricted to ${\mathcal M}$, plus
global boundedness (GB) guarantees GAS of the
origin of the cascade. The extension discussed in this note is threefold.
First, we consider nonlinear hybrid systems. Then, 
we extend this result to the case where we do not necessarily
insist on the origin and, finally, we do not insist on a cascaded nature of the overall system. 

The result finds applications in several contexts. The authors of this note have been using it as a tool to show stability properties in systems where different substates converge to desirable manifolds and hierarchically reach a desirable final closed set \cite{BisoffiCDC14,SassanoACC12,TregouetTCST15,CordioliECC15}. Several additional application can be found in the literature, for example it may be used as an alternative way to prove the results in \cite{OttAuto15}.

{\bf Notation}. 
$\ball$ is the open unit ball centered at the origin. We denote 
$\mathcal{X}+\mathcal{Y} = \{z: z = x+y, x \in \mathcal{X}, y \in \mathcal{Y} \}$. We denote the distance $|z|_{\mathcal{M}}$ of a
    point $z$ from the set $\mathcal{M}$ as $|z|_{\mathcal{M}} :=
    \inf_{w\in \mathcal{M}} |z-w|$.


\section{Preliminaries and Definitions}
\label{sec:preliminaries}

Using the formalism in \cite{TeelBook12}, we consider a hybrid dynamical system described by
\begin{eqnarray}
\label{eq:hybrid}
{\mathcal H}: \quad \left\{
\begin{array}{lr}
x \in C, \quad \dot x \in F(x),\\
x \in D, \quad x^+ \in G(x). 
\end{array}
\right.
\end{eqnarray}
We suppose here that $(C,F,D,G)$ satisfy the hybrid basic conditions given in
\cite[Ass. 6.5]{TeelBook12}, which are reported here for completeness.

\begin{assumption}(Hybrid basic conditions)
\label{ass:hybrid_basic}
\begin{enumerate}[(i)]
\item $C$ and $D$ are closed subsets of $\mathbb{R}^{n}$;
\item $F : \mathbb{R}^{n} \rightrightarrows \mathbb{R}^{n}$ is outer semicontinuous and locally
bounded relative to $C$, and it is non-empty and convex for each $x \in C$;
\item $G : \mathbb{R}^{n} \rightrightarrows \mathbb{R}^{n}$ is outer semicontinuous and locally
bounded relative to $D$, and it is non-empty for each $x \in D$.
\end{enumerate}
\end{assumption}

\begin{remark}
The conditions in Assumption \ref{ass:hybrid_basic} are sufficient to imply nominal and robust well-posedness
of system (\ref{eq:hybrid}), see \cite[Def. 6.2]{TeelBook12}. Note that special cases correspond to continuous-time differential equations $\dot x = f(x)$ where $f$ is continuous (this corresponds to the case $C=\real^n$ and $D=\emptyset$) and discrete-time difference equations $x^+ = g(x)$ where $g$ is continuous (this corresponds to the case $D=\real^n$ and $C=\emptyset$).
In these two special cases, all the presented results apply and the solution concept introduced in \cite{TeelBook12} reduces to the classical solution concepts for continuous-time and discrete-time systems, respectively.
\end{remark}

We begin the section by introducing some definitions necessary to make
the statements above more precise. 

\begin{definition} \label{def:stab+attr}
Given hybrid system ${\mathcal H}$ in  (\ref{eq:hybrid}), a 
 closed set $\mathcal{X}\subset \real^n$ is\\
\begin{enumerate}
\item \label{itdef:1}
 {\em stable} for (\ref{eq:hybrid}) if for each pair $\epsilon>0$,
  $\Delta>0$, there exists $\delta>0$ such that all solutions $\varphi$ to ${\mathcal H}$ satisfy:
$$
\varphi(0,0) \in ((\mathcal{X}+\delta \ball) \cap \Delta \ball) \;\Rightarrow \;\varphi(t,j) \in \mathcal{X}+\epsilon \ball, 
$$ 
\noindent for all $(t,j) \in \dom \varphi$;

\item 
{\em attractive} if there exists
  $\delta>0$ such that all solutions $\varphi$ to ${\mathcal H}$ satisfy
\begin{align}
\label{eq:conv}
\varphi(0,0) \in (\mathcal{X}+\delta \ball) \;\Rightarrow \; \lim\nolimits\limits_{t+j \to +\infty} |\varphi(t,j)|_{\mathcal{X}} = 0;
\end{align}

\item  (locally) {\em asymptotically stable} (AS) if
  it is stable and attractive. Moreover, its {\em basin of attraction} $B_{\mathcal{X}}$
  is the largest set of initial conditions from which all solutions to 
  (\ref{eq:hybrid})
  converge to $\mathcal{X}$;

\item 
 {\em globally attractive} if
  (\ref{eq:conv}) holds for all $\delta>0$;

\item
 {\em strongly forward invariant} if all solutions starting in ${\mathcal{X}}$ remain in ${\mathcal{X}}$ for all times.
\end{enumerate}

A compact set $\mathcal{X}_c\subset \real^n$ is
\begin{enumerate}

\item 
 {\em uniformly attractive} from a compact set $K \subset \real^n$, $K \supset \mathcal{X}_c$,  if 
for each $\epsilon>0$, there exists $T$ such that 
all solutions $\varphi$ to ${\mathcal H}$ satisfy
$$
\varphi(0,0) \in K \;\Rightarrow \;  |\varphi(t,j)|_{\mathcal{X}_c} \leq \epsilon,\quad
$$
\noindent for all $(t,j) \in \dom \varphi$ such that $t+j\geq T$;

\item 
 {\em uniformly (locally) asymptotically stable} (UAS) if
  it is stable and uniformly attractive from each compact subset of
  its basin of attraction $B_{\mathcal{X}_c}$;

\item 
 {\em uniformly globally asymptotically stable} (UGAS) if it
  is UAS with $B_{\mathcal{X}_c} = \real^n$.
\end{enumerate}

Given a closed forward invariant set ${\mathcal Y}$ for the hybrid system
(\ref{eq:hybrid}),
each one of
the above properties holds {\em relative to the set ${\mathcal Y}$} if
it holds for initial conditions restricted to ${\mathcal Y}$.
\end{definition}

\begin{definition}
\label{def:boundedness}
Given system ${\mathcal H}$ in (\ref{eq:hybrid}) and 
a compact set $K \subset \real^n$,
solutions (or trajectories) of ${\mathcal H}$ are {\em uniformly bounded
  from $K$} if there exists $\Delta>0$ such that all solutions $\varphi$
  are such that $\varphi(0,0) \in K$ implies $|\varphi(t,j)| < \Delta$ for
all $(t,j) \in \dom \varphi$. Moreover, given a subset ${\mathcal X}$ of $\real^n$, 
solutions (or trajectories) of ${\mathcal H}$  are {\em uniformly bounded
  from ${\mathcal X}$} if they are uniformly bounded from each compact
subset of ${\mathcal X}$. If ${\mathcal X} = \real^n$ then
trajectories are {\em uniformly globally bounded} (UGB).
\end{definition}

\begin{remark}
The UGAS property in Definition~\ref{def:stab+attr} does not explicitly require
the UGB property of Definition~\ref{def:boundedness}. This is because
${\mathcal{X}_c}$ compact and
$f$ continuous, together with local stability of ${\mathcal{X}_c}$ and
global convergence to ${\mathcal{X}_c}$ implies uniform global
boundedness (see \cite[Prop. 6.3]{Goebel06}). Nevertheless, it has
been shown in \cite{TeelAuto06} that for time-varying nonlinear
systems (in other words, unbounded attractors) local stability and uniform global convergence does not guarantee UGB. In turns, UGB is necessary to
enforce a bound on the overshoots of the trajectories, thereby being
able to guarantee a $\mathcal{KL}$ bound (see \cite[Lemma 4.5]{Khalil3} or
\cite[Thm 7.12]{TeelBook12} for equivalence between UGAS and existence
of a global $\mathcal{KL}$ bound). 
\end{remark}

\section{Main stability result}
\label{sec:main}

The following lemma extends the result
of \cite{Seibert90,Sontag89CDC} on cascaded systems. 
Some related results
 have been presented 
for the continuous-time case 
 in \cite{IggidrMCSS96,SeibertSPRINGER70,Seibert95} and
 their relation with the lemma are clarified next.

\begin{lemma}
\label{lem:new}
Consider hybrid system ${\mathcal H}$ in (\ref{eq:hybrid}) satisfying the basic conditions in Assumption~\ref{ass:hybrid_basic}, and
assume that:
\begin{enumerate}

\item
\label{it:0}
the closed set ${\mathcal M}_e \subset \real^n$ is strongly forward invariant;

\item
\label{it:1}
the closed set ${\mathcal M}_i \subset {\mathcal M}_e$ is stable and globally attractive relative to ${\mathcal M}_e$;

\item\label{it:2}
 the compact set ${\mathcal M}_{\circ}\subset {\mathcal M}_i$ is stable
 and globally attractive relative to ${\mathcal M}_i$.
\end{enumerate}
Then the set ${\mathcal M}_{\circ}$ is UAS relative to ${\mathcal M}_e$
for ${\mathcal H}$, with
basin of attraction $B_{\mathcal{M}_{\circ}}$ such that $B_{\mathcal{M}_{\circ}} \cap
{\mathcal M}_e$
coincides with the largest subset of ${\mathcal M}_e$ from which all
solutions are bounded. In particular, if solutions are UGB relative to ${\mathcal M}_e$,
then the set ${\mathcal M}_{\circ}$ is UGAS
for ${\mathcal H}$ relative to ${\mathcal M}_e$.
Finally, if the above stated attractivity properties only hold
locally, then the set ${\mathcal M}_{\circ}$ is UAS relative to $\mathcal{M}_e$.
%
%
\end{lemma}

\begin{proof}
First notice that the property at item~\ref{it:2} is well
defined because, from item~\ref{it:1}, the set ${\mathcal M}_i$ is
(strongly) forward invariant.

\noindent
{\em Preliminary Step.}
The lemma is proven using \cite[Corollary 19]{GoebelCSM09}, in a
similar way to what is done around \cite[eqs. (23),
(24)]{GoebelCSM09}. 
%
%
In particular, we apply
\cite[Corollary 19]{GoebelCSM09} by intersecting ${\mathcal M}_i$
with arbitrarily large compact subsets of ${\mathcal M}_e$. Namely,
given an arbitrary positive number $\bar M$, we apply 
\cite[Corollary 19]{GoebelCSM09} to hybrid system 
$\overline{\mathcal H}:=(\overline C, \overline F, \overline D, \overline G)$ 
and compact sets ${\mathcal A}_1$ and ${\mathcal A}_2$ selected as:\footnote{Note that it is not necessary to intersect $\overline G$ with ${\mathcal M}_e$ due to the assumed strong forward invariance property.}
\begin{equation}
\label{eq:sets}
\begin{array}{l}
\overline C:= C \cap {\mathcal M}_e \cap ({\mathcal M}_{\circ} + \bar M\ball) \\
\overline F := F \\
\overline D:= D \cap {\mathcal M}_e \cap ({\mathcal M}_{\circ} + \bar M\ball) \\
\overline G := G \cap ({\mathcal M}_{\circ} + \bar M\ball), \\
\mathcal{A}_1 = {\mathcal M}_i\cap (\overline C\cup \overline D) \\ \mathcal{A}_2 = {\mathcal M}_{\circ}.
\end{array}
\end{equation}
 From \cite[Corollary 19]{GoebelCSM09} 
 we conclude that the set ${\mathcal M}_{\circ}$ is globally
asymptotically stable (GAS) for the restricted hybrid dynamics $\overline {\mathcal H}$.
%

\noindent
{\em UAS of ${\mathcal M}_{\circ}$ relative to ${\mathcal M}_e$}. 
Recalling
that (global) AS of ${\mathcal M}_{\circ}$
comprises its stability and noticing that,
from (\ref{eq:sets}),
 ${\mathcal M}_{\circ}$ 
is in the interior of $\overline C\cup \overline D$ relative to 
${\mathcal M}_e \cap (\overline C\cup \overline D)$, 
then for any scalar $\bar M>0$, stability of ${\mathcal M}_{\circ}$ 
for $\overline{\mathcal H}$,
together with strong forward invariance of ${\mathcal M}_e$,
implies that there is a small enough $\delta>0$
such that all solutions to $\overline{\mathcal H}$ starting in 
$({\mathcal M}_{\circ}+\delta \ball) \cap {\mathcal M}_e$
are contained in
$({\mathcal M}_{\circ}+\overline M \ball) \cap {\mathcal M}_e$,
namely all solutions to ${\mathcal H}$ starting in this set are also solutions to 
the restricted dynamics $\overline{\mathcal H}$.
Then AS of ${\mathcal M}_{\circ}$ for $\overline{\mathcal H}$
relative to ${\mathcal M}_e$ implies 
AS of ${\mathcal M}_{\circ}$ for ${\mathcal H}$
relative to ${\mathcal M}_e$.
 Moreover, since ${\mathcal M}_{\circ}$
is compact and ${\mathcal H}$ satisfied Assumption~\ref{ass:hybrid_basic}, then \cite[Prop. 6.2]{Goebel06} (see
also \cite[Thm 7.21]{TeelBook12}) implies that AS of ${\mathcal
  M}_{\circ}$ is uniform in ${\mathcal M}_e$.

\noindent
{\em Basin of attraction}. 
First
 note that for any bounded solution in ${\mathcal M}_e$ there exists a large enough selection of
$\bar M$ such that $\overline C \cup \overline D \cup \overline G(\overline D)$ 
in (\ref{eq:sets}) contains that solution (and 
\cite[Corollary 19]{GoebelCSM09} 
establishes its convergence). Therefore, any such solution is guaranteed to be
contained in the set $B_{\mathcal{M}_{\circ}}\cap  {\mathcal M}_e$.
Conversely, any solution which is not (uniformly) bounded is not in the
domain of attraction because 
\cite[Prop. 6.3]{Goebel06} establishes that UAS of
a compact set $\mathcal{M}_{\circ}$ implies that solutions are uniformly
bounded from $B_{\mathcal{M}_{\circ}}$.
As a consequence, by the arbitrariness of
$\bar M$, set $B_{\mathcal{M}_{\circ}}\cap  {\mathcal M}_e$
  coincides with the largest subset of ${\mathcal M}_e$
from which all solutions are bounded.
\end{proof}

\begin{remark} \label{rem:KL}
Applying \cite[Thm 6.5]{Goebel06} we have that
both in the local and global cases
 the result of Lemma~\ref{lem:new} implies the existence of a
class $\mathcal{KL}$ estimate for the distance of the solution from the attractor
${\mathcal M}_{\circ}$.  In particular, when 
${\mathcal M}_e = \real^n$ and
all solutions are bounded,
so that the lemma implies UGAS (namely
$B_{\mathcal{M}_{\circ}} = \real^n$), there exists $\beta \in \mathcal{KL}$ such that all solutions satisfy
$|\varphi(t,j)|_{{\mathcal M}_{\circ}} \leq \beta(|\varphi(0,0)|_{{\mathcal{M}_{\circ}}},t+j)$ for all $(t,j) \in \dom \varphi$. When $B_{\mathcal{M}_{\circ}}$ is a
strict subset of $\real^n$ the bound holds with the distance $|\cdot
|_{{\mathcal{M}_{\circ}}}$ replaced by a proper indicator function
of $\mathcal{M}_{\circ}$ with respect to $B_{\mathcal{M}_{\circ}}$.
Finally, \cite[Prop. 6.4]{Goebel06}
implies that $B_{\mathcal{M}_{\circ}}$ is an open subset of $\real^n$ containing $\mathcal{M}_{\circ}$.
\end{remark}



\begin{remark}
The result in Lemma~\ref{lem:new} is relevant mostly due to its
simple formulation and broad applicability. From the point of view of the proof technique,
it is a simple extension of existing results. 
For the case ${\mathcal M}_e = \real^n$ (so that item~\ref{it:0} is trivially guaranteed) and continuous-time systems $\dot x =f(x)$,
an alternative proof 
can be obtained using the results in \cite{IggidrMCSS96} if one
adds the extra assumption that $f$ is locally Lipschitz
and the non-uniform attractivity property in (\ref{eq:conv}) is 
strengthened to a uniform one.
In this strengthened case,
applying
the converse Lyapunov theorem in \cite{TeelPraly00} (see also \cite{Wilson69}),
item~\ref{it:1} implies the existence of a smooth
positive semidefinite Lyapunov function $V_{\mathcal M}$ establishing global asymptotic stability of
${\mathcal M}_i$ and then we can apply
\cite[Corollary 1]{IggidrMCSS96} with ${\mathcal M} = M^* = M = M_0$.
Indeed, using the notation in
\cite{IggidrMCSS96} we have $M_0 = {\mathcal M}$ (the set where
$V_{\mathcal M}$ vanishes) and from the converse Lyapunov construction, also
$M=M_0$. Finally, since $M= M_0$ is invariant, then $M^* = M$.
The reader is also referred to \cite[Example 2]{IggidrMCSS96} where
the same idea is applied to a cascaded interconnection.
The
local asymptotic stability (AS) property of Lemma~\ref{lem:new}, under the
stated continuity assumptions for $f$, has also been established (without any proof) in
\cite[Theorem 2]{SeibertSPRINGER70}. A proof can be found in the more
general case of (possibly infinite dimensional) semidynamical systems
in \cite[Theorem 4.13]{Seibert95}.
Finally, under a strengthened local Lipschitz
assumption on $f$, the proof technique of
\cite[Example 2]{IggidrMCSS96} can be used
to establish the local result.
\end{remark}

\begin{remark}
One may wonder whether the result of Lemma~\ref{lem:new} is truly
more general than the classical cascaded systems result of
\cite{Sontag89CDC}, \cite[Example 2]{IggidrMCSS96} and
\cite[eqs. (23), (24)]{GoebelCSM09} (at least for the continuous-time case). A partial answer to this question arises if one
focuses on a purely continuous-time setting and on the case ${\mathcal M}_{i} = \{x :
h(x)=0\}$ and attempts to prove the existence of 
a nonlinear change of coordinates
highlighting a cascaded structure between an upper subsystem whose
state is given by $x_1 = h(x) \in \real^{n_y}$ and whose dynamics is guaranteed to
converge to zero by item~\ref{it:1} of Lemma~\ref{lem:new}, and a
lower subsystem whose state should be given by the completion $x_2 =
h_{compl}(x)\in \real^{n-n_y}$, where the function $h_{compl}$ is such that the
overall function $x \mapsto T(x) := [h(x)\; h_{compl}(x)]$ is a
diffeomorphism.
Unfortunately, assessing whether such change of coordinates exists does not seem
to be an easy task to accomplish, in general. 
Indeed, even if only wanting to
define this change of coordinates locally in a 
neighborhood $U^{o}$ of a point $x^{o}\in \real^n$, its existence is
related to invariance of the \emph{completely integrable} distribution
generated by the columns of $Ker(\nabla h(x)^{T})$ with respect
to the vector field $f$, which is a much \emph{stronger}
hypothesis of invariance of the submanifold $\mathcal{M}_i$
with respect to $f$, see \cite[Lemma 1.6.1]{Isidori}
for more detailed discussions.
\end{remark}

\noindent
{\bf Acknowledgments}.
The authors would like to thank Andy Teel for useful suggestions.

\bibliographystyle{plain}


\end{document}